\documentclass{amsart}
\usepackage{amsfonts}
\usepackage{amsmath}
\usepackage{amssymb}
\usepackage{amsthm}
\usepackage{color}
\usepackage{comment}

 \usepackage[utf8]{inputenc} 
\usepackage{pdfpages}

\newtheorem{Thm}{Theorem}
\newtheorem{Lem}[Thm]{Lemma}

\newtheorem{Cor}[Thm]{Corollary}

\theoremstyle{definition}
\newtheorem{Def}[Thm]{Definition}
\newtheorem{Ex}[Thm]{Example}

\theoremstyle{remark}
\newtheorem{Rem}[Thm]{Remark}

\numberwithin{equation}{section}

\def\RR{{\mathbb{R}}}

\DeclareRobustCommand{\rchi}{{\mathpalette\irchi\relax}}
\newcommand{\irchi}[2]{\raisebox{\depth}{$#1\chi$}} 

\DeclareMathOperator{\Ima}{Im}

\title[How to cut a cake with a Gram matrix ]{How to cut a cake with a Gram matrix}
\date{\today}

\author[G.~Ch\`eze]{Guillaume Ch\`eze}
\address{Guillaume Ch\`eze: Institut de Math\'ematiques de Toulouse\\
Universit\'e Paul Sabatier Toulouse 3 \\
MIP B\^at. 1R3\\
31 062 TOULOUSE cedex 9, France}
\email{guillaume.cheze@math.univ-toulouse.fr}

\author[L.~Amodei]{Luca Amodei}
\address{Luca Amodei: Institut de Math\'ematiques de Toulouse\\
Universit\'e Paul Sabatier Toulouse 3 \\
MIP B\^at. 1R3\\
31 062 TOULOUSE cedex 9, France}
\email{luca.amodei@math.univ-toulouse.fr}

\date{\today}

\begin{document}
	
\maketitle

\begin{abstract}
In this article we study the problem of fair division. In particular we study a notion introduced by J. Barbanel that generalizes super envy-free fair division.
We give a new proof of his result. Our approach allows us to give an explicit bound for this kind of fair division.\\
Furthermore, we also give a theoretical answer to an open problem posed by Barbanel in 1996. Roughly speaking, this question is: how can we decide if there exists a fair division satisfying some inequalities constraints?\\
Furthermore, when all the measures are given with piecewise constant density functions then we show how to construct effectively such a fair division.
\end{abstract}


\section*{Introduction}
In the following $X$ will be a measurable set. This set represents an heterogeneous good, e.g. a cake, that we want to divide between n players. A division of the cake is a partition $X=\sqcup_{i=1}^n X_i$, where each $X_i$ is a measurable subset of $X$. After the division $X_i$ is given to the $i$-th player. A natural and old problem is: how to get a fair division?\\
This  problem  appears  when  we study  division  of  land,  time  or  another divisible resource  between  different  agents  with different points of view. These  problems appear in the economics, mathematics, political
science, artificial intelligence and computer science literature, see \cite{Moulin,handbook}.\\
In order to study this problem, to each player is associated a non-atomic probability measure $\mu_i$. Thus, in particular $\mu_i(X)=1$, and $\mu_i(A \sqcup B)=\mu_i(A) +\mu_i(B)$, where $A$, and $B$ are disjoint measurable sets. These measures represent the preference of each player. Severall notions of fair divisions exist:
\begin{itemize}
 \item Proportional  division: $\forall i,\, \mu_i(X_i) \geq 1/n$.
 \item Exact  division: $\forall i,\, \forall j,\, \mu_i(X_j)=1/n$.
 \item Equitable division: $\forall i,\, \forall j,\, \mu_i(X_i) = \mu_j(X_j)$.
 \item Envy-free division: $\forall i,\, \forall j,\, \mu_i(X_i)\geq \mu_i(X_j)$.
  \end{itemize}

 All these fair divisions are possible, see e.g \cite{Steinhaus,DubinsSpanier,BramsTaylor,RobertsonWebb,Cheze,Segal-Halevi,Weller}.\\
%
  Some fair divisions are possible under some conditions.
Barbanel has shown in \cite{Barbanel94} that a super envy-free division is possible if and only if the measures $\mu_i$ are linearly independent. We recall the definition of a super envy-free fair division:
 \begin{itemize}
\item Super envy-free division:  $\forall i,\, \forall j\neq i,\, \mu_i(X_i)>1/n> \mu_i(X_j)$.
\end{itemize}
 Actually, if the measures are linearly independent then there exists a real $\delta > 0$ such that:
$$\mu_i(X_i)\geq 1/n+\delta, \textrm{ and } \mu_i(X_j)\leq 1/n-\delta/(n-1).$$
We can define an even more demanding fair division. For example, we can imagine that the first player would like to get a partition such that:\\
$\mu_1(X_1)=1/n+ 3 \delta$, \\
 $\mu_1(X_3)=1/n+2 \delta$,\\
   $\mu_1(X_4)= \mu_1(X_5) = 1/n+\delta$,\\
    $\mu_1(X_2)=\mu_1(X_6)=1/n-6\delta$.\\
    
This means that the third, the forth and the fifth player are friends with the  first player, but the second and sixth player are not friends with this player. Furthermore, the first player prefers the third to the forth and the fifth player. We can also imagine that the other players have also preferences between the other players.
These kinds of conditions have also been studied by Barbanel in \cite{Barbanelarticle,Barbanelbook}. This leads to a notion of fair division that we call \emph{hyper envy-free}.\\
\begin{Def}
Consider a matrix $K=(k_{ij}) \in M_n(\RR)$, such that for all $i=1,\ldots,n$, $\sum_{j=1}^n k_{ij}=0$ and a point $p=(p_1,\ldots,p_n)$ such that $\sum_{j=1}^np_j=1$, with $p_i \geq 0$.\\
We say that a partition $X=\sqcup_{i=1}^n X_i$ is \emph{hyper envy-free} relatively to $K$ and $p$  when there exists a real number $\delta>0$ such that 
$$\mu_i(X_j)= p_j + k_{ij} \delta.$$
\end{Def}

For example, for a super envy-free division we have $k_{ij}=-1/(n-1)$,  $k_{ii}=1$ and $p=(1/n,\ldots,1/n)$.\\

Barbanel has given a criterion for the existence of an hyper envy-free division. Unfortunately, the proof of this result does not give an explicit bound on $\delta$ and a natural question is: \emph{How big $\delta$ can be?}\\

In this paper we give a new proof of Barbanel's result. Furthermore, our strategy generalizes the approach for computing super envy-free fair division given by Webb in \cite{Webb}. This allows us to give a bound on $\delta$ in terms of the measures.\\

Another question related to fair division and asked by Barbanel is the following, see \cite{Barbanelarticle}:\\

\emph{Suppose that $p=(p_1,\ldots,p_n)$ is a point such that $p_1+\cdots+p_n=1$ with $p_i$ positive and $r_{ij}$ are $n^2$ relations in $\{ <, = , > \}$. How can we decide if there exists a partition $X=\sqcup_{i=1}^n X_i$ such that $\mu_i(X_j) \, r_{ij}\, p_j$?\\
}

This problem is also a generalization of the super envy-free fair division problem. Indeed, super envy-free fair division corresponds to the situation where $p$ is $(1/n,\ldots,1/n)$, $r_{ii}$ is the relation ``$>$" and $r_{ij}$ is ``$<$".\\
We give a theoretical answer to Barbanel's question in the last section of this article.\\

The organization of our article is the following. In the next section we present our toolbox. We recall the Dvoretzky, Wald, Wolfowitz's theorem that will be the main ingredient of our proof. In Section \ref{sec:existence}, we prove the existence of an hyper envy-free division under some linear conditions on the measures $\mu_i$.  This gives a new proof of Barbanel's theorem. A direct consequence of our construction gives a bound on $\delta$. At last, in Section \ref{sec:test} we solve the open question asked by Barbanel. Furthermore, when all the measures are given by piecewise constant density functions we show how to construct effectively an hyper envy-free fair division. This gives a method to construct a partition such that $\mu_i(X_j) \, r_{ij} \, p_j$.\\

\section{Our toolbox}\label{sec:toolbox}
\begin{Def}
A matrix $M=(\mu_i(X_j))$ is said to be a \emph{sharing matrix} when \mbox{$X=\sqcup_{i=1}^nX_i$} is a partition of $X$.
\end{Def}

A sharing matrix is a \emph{row stochastic matrix}, this means that each coefficients are nonnegative and  the sum of the coefficients of each row is equal to 1. In the following, we will use classical results about this kind of matrices. We recall below without proofs these results.

\begin{Lem}\label{lem:sto}

We denote by $e$ the vector $(1,\ldots,1)^T$. If $S=(s_{ij})$ and $T=(t_{ij})$ are matrices, and $s, t \in \RR$ are such that  $S \, e=s \, e$ and $T \, e=t \, e$, then:
\begin{enumerate}
\item  $N=S \, T$ is a matrix such that $N \, e=st \, e$.
\item If $S$ is invertible,  then $s\neq 0$ and $S^{-1}$ is such that $S^{-1} \, e= \frac{1}{s} \, e$.
\end{enumerate}
\end{Lem}

In the following we define a new measure. This measure will be usefull to write each $\mu_i(X_j)$ in term of the same measure.

\begin{Def}
We denote by $\mu$ the measure $\mu=\mu_1+\cdots+\mu_n$.\\
The Radon-Nikodym derivative $d\mu_i/d\mu$ is denoted by $f_i$.
\end{Def}

The Radon-Nikodym derivative $f_i$ exists because $\mu_i$ is absolutely continuous relatively to $\mu$. By definition we have for any measurable subset $A \subset X$:
$$\mu_i(A)=\int_A f_i(x) d\mu.$$

\begin{Def}
The Dvoretzky, Wald, Wolfowitz set (DWW set) is the set of all matrices 
$$\Big(\int_X\eta_j(x) d\mu_i\Big)$$
where $\eta_1,\ldots,\eta_n$,  are positive functions such that $\eta_1(x)+\cdots+\eta_n(x)=1$.
\end{Def}

Thanks to the previous notations a matrix in the DWW set can be written
$$\Big(\int_X\eta_j(x) f_i(x)d\mu\Big).$$

The following deep result will be one of the main ingredient of our proof.

\begin{Thm}[Dvoretzky, Wald, Wolfowitz \cite{DWW}]
The DWW set is the set of sharing matrices.
\end{Thm}

This theorem is classical when we study fair division problems. As a corollary we get that the set of sharing matrices is a convex set. Barbanel's proof uses only this corollary.\\

In his book Barbanel gives the following lemma, see \cite[Lemma 9.10]{Barbanelbook}.
\begin{Lem}\label{lem:prop_f_i}
The functions $f_i$ satisfy:
$$f_1(x)+\cdots + f_n(x)=1,$$
$$0\leq f_i(x) \leq 1, \, \forall x \in X.$$
\end{Lem}

We can then construct a sharing matrix thanks to the density functions $f_i$. The following Gram matrix will be the basis of our construction.

\begin{Cor}
The Gram matrix $$G(f_1,\ldots,f_n)=\Big(\int_Xf_j(x) f_i(x)d\mu\Big),$$
is a sharing matrix.
\end{Cor}
 
 We remark that the matrix $G(f_1,\ldots,f_n)$ is well defined since by Lemma \ref{lem:prop_f_i}, $f_i \in L^{\infty}(X)$.\\

Thanks to the Dvoretsky, Wald, Wolfowitz's theorem we know that $G(f_1,\ldots,f_n)$ is a sharing matrix. Thus, there exists a partition $X=\sqcup_{i=1}^n X_i$ such that 
$$\int_X f_j(x)f_i(x)d\mu=\mu_i(X_j).$$

\begin{Rem} \label{rem_prop} The matrix $G(f_1,\ldots,f_n)$  gives a proportional fair division. Indeed, we apply Cauchy-Schwartz's inequality in $L^2(X)$ where 
$$<f,g>=\int_X f(x)g(x)d\mu,$$
 and we get $<1,f_i> \leq \|f_i\| \, \|1\|$.\\
However 
$$\|1\|^2=\int_X d\mu=\int_X d\mu_1+\cdots+ d\mu_n=n$$
 and $$<1,f_i>=\int_X f_i d\mu=\mu_i(X)=1.$$
  This gives $\| f_i\|^2 \geq 1/n$. Thus all diagonal elements of  $G(f_1,\ldots,f_n)$ satisfy 
  $$\|f_i\|^2=\mu_i(X_i) \geq 1/n.$$
   This gives a proportional fair division.\\

\end{Rem}

This matrix has also other usefull properties. For example, this matrix is symmetric. Now, we recall a classical  property of a Gram matrix.  

\begin{Lem}\label{lem:nullspace}
Let  $x$ denote the vector $x=(x_1, \ldots , x_n)^T$, and $\ker(G)$ the kernel of a matrix $G$.  We have the 
equivalence 
$$ x \in \ker\big(G(f_1,\ldots,f_n)\big) \iff \sum_{i=1}^n x_i f_i  = 0  \textrm{ in } L^2(X) \iff \sum_{i=1}^n x_i \mu_i  = 0.$$
\end{Lem}
\begin{proof} 
We have: $$G(f_1,\ldots,f_n)\,x=\big(<f_1,\sum_{j=1}^n x_jf_j >,\ldots,<f_n,\sum_{j=1}^n x_jf_j >\big)^{T}.$$
 Thus 
 $x \in \ker\big(G(f_1,\ldots,f_n)\big)$ means $<  f_i , \sum_{j=1}^n x_jf_j > =0, \forall i=1, \ldots, n.$ Combining these equalities we get $\| \sum_{j=1}^n x_j f_j \|^2 = 0$ which implies $\sum_{j=1}^n x_j f_j  = 0$ in $L^2(X)$. The converse implication is straightforward. \\
 Furthermore, we have the following equivalence:
\begin{eqnarray*} 
 \sum_{i=1}^n x_i f_i  = 0  \textrm{ in } L^2(X) &\iff & \textrm{ for all measurable set  } A, \quad \int_A \sum_{i=1}^n x_i f_i(x)d\mu=0\\
 &\iff &\textrm{ for all measurable set  } A, \quad \sum_{i=1}^n x_i \int_A f_i(x)d\mu =0 \\
 &\iff &\textrm{ for all measurable set  } A, \quad \sum_{i=1}^n x_i \mu_i(A)=0. 
 \end{eqnarray*}
 This gives the desired result.
\end{proof}

In terms of the measures $\mu_i$ this property implies that $G(f_1,\ldots,f_n)$ is nonsingular if and only if the measures $\mu_i$ are independent. \\

\begin{Lem}\label{lem:G+}
We denote by $G(f_1,\ldots,f_n)^{+}$ the Moore-Penrose pseudo-inverse of $G(f_1,\ldots,f_n)$ and  by $e$ the vector $e=(1,\ldots,1)^T$. We have
$$G(f_1,\ldots,f_n)^{+} \, e=e.$$
\end{Lem}
\begin{proof}
For a matrix $G$ we have the following classical property, see e.g. \cite[Ex. 25, p.115]{Greville}: 
$$(\star) \quad G^{+}=\lim\limits_{\substack{\rho \to 0 \\ \rho>0}} G^T(GG^T+\rho I)^{-1}.$$ As $G(f_1,\ldots,f_n)$ is symmetric and row stochastic, then  $G(f_1,\ldots,f_n)^{T}$ is also row stochastic. Thus, by Lemma \ref{lem:sto}, the product $G(f_1,\ldots,f_n) \, G(f_1,\ldots,f_n)^T$ is a row stochastic matrix. Then, 
$$\big(G(f_1,\ldots,f_n) \, G(f_1,\ldots,f_n)^T + \rho \,  I \big) \, e=(1+\rho) \,  e.$$
Since $\rho > 0$,  the matrix 
$\big( G(f_1,\ldots,f_n) \, G(f_1,\ldots,f_n)^T + \rho \,  I \big)$ is nonsingular. Indeed, $G(f_1,\ldots,f_n) \, G(f_1,\ldots,f_n)^T$ is a semi-definite positive matrix and thus the matrix $G(f_1,\ldots,f_n)\, G(f_1,\ldots,f_n)^T + \rho \,  I$ is definite positive. Using again Lemma \ref{lem:sto},  we deduce that $\big(G(f_1,\ldots,f_n) \, G(f_1,\ldots,f_n)^T+ \rho   \, I\big)^{-1} \,e=\frac{1}{1+\rho} \, e.$  
 Since $G(f_1,\ldots,f_n)^{T}$ is row stochastic, then property $(\star)$ gives the desired conclusion.
\end{proof}

The next lemma shows how to get new sharing matrices. This idea was already present in Webb's algorithm for computing a super envy-free fair division, see \cite{Webb}.

\begin{Lem}\label{lem:stab-sto}
Let $M \in M_n(\RR)$ be a sharing matrix and $S \in M_n(\RR)$ be a row stochastic matrix. Then $M\,S$ is a sharing matrix.
\end{Lem}

\begin{proof}
Suppose that the sharing matrix is given by $M=\big(\mu_i(X_j)\big)$ and  denote by $S=(s_{ij})$ the stochastic matrix. We set $\eta_j=\sum_{k=1}^n s_{kj}\rchi_{X_k}$, where $\rchi_{X_k}$ is the indicator function associated to $X_k$.\\
We remark that $\eta_j(x)$ is positive and $\eta_1(x)+\cdots+\eta_n(x)=1$ because $X=\sqcup_{i=1}^n X_i$ is a partition and $S$ is a row stochastic matrix.\\
Furthermore, we have 
$$\int_X \eta_j(x)f_i(x) d\mu = \sum_{k=1}^n s_{kj} \int_X \rchi_{X_k}(x) f_i(x) d\mu
=\sum_{k=1}^n \mu_i(X_k) s_{kj} 
.$$
This gives: $\big( \int \eta_j(x)f_i(x) d\mu \big)_{ij}=M\,S$. Then by the Dvoretzky, Wald, Wolfowitz's theorem $M\,S$ is a sharing matrix.
\end{proof}

\section{Existence of hyper envy-free division and an explicit bound}\label{sec:existence}
\subsection{Sufficient conditions}\label{subsec:cond-suff}

\begin{Def}\label{def:propre}
A matrix $K=(k_{ij})$ is said to be proper relatively to the measures $\mu_1$, \ldots, $\mu_n$ when the two following conditions hold:
\begin{enumerate}
\item for all $i=1,\ldots,n$, $\sum_{j=0}^n k_{ij}=0$,
\item if $\sum_{i=0}^n \lambda_i \mu_i=0$, where $\lambda_i \in \RR$, then  $\sum_{i=0}^n \lambda_i k_{ij}=0$, for all $j=1,\ldots,n$.
\end{enumerate}
\end{Def}

With this definition we can state our extension of Barbanel's theorem. 

\begin{Thm}\label{thm1}
If we denote by $P$ the matrix where each row is $(p_1, \ldots,p_n)$, and $\sum_{i=1}^n p_i=1$, with $p_i > 0$, and $K$ is a proper matrix, then for

$$\displaystyle 0\leq \delta \leq \mathcal{B}:=\dfrac{ \underset{i}{\min}(p_i) }{\underset{ij}{\max}|\big(G(f_1,\ldots,f_n)^{+} K\big)_{ij}|}$$

we have:
\begin{enumerate}
\item $G(f_1,\ldots,f_n)^{+}\,(P+\delta K)$ is a row stochastic matrix.
\item $G(f_1,\ldots,f_n)\,G(f_1,\ldots,f_n)^{+}\,(P+\delta K)=(P+\delta K)$.
\item $P+\delta K$ is a sharing matrix, i.e. there exists an hyper envy-free fair division relatively to $K$ and $p$, with $\delta \leq \mathcal{B}$.
\end{enumerate} 
\end{Thm}
The last item is an effective statement of Barbanel's theorem. Our approach allows us to give a bound $\mathcal{B}$ on $\delta$.\\
The case $\delta=0$ is classical and follows form the convexity of the set of sharing matrices. 

\begin{proof}
By Lemma \ref{lem:G+} we have $G(f_1,\ldots,f_n)^{+} e= e.$ Using  Lemma \ref{lem:sto} and the properties of $P$ and $K$ we get $$G(f_1,\ldots,f_n)^{+}  (P+\delta K) \, e=e.$$
In order to prove the first item we now have to show that the coefficients of $G(f_1,\ldots,f_n)^{+} \,  (P+\delta K)$ are non-negative for $\delta$ sufficiently small.\\
We set $G(f_1,\ldots,f_n)^{+}=(g^{+}_{ij})$ and $K=(k_{ij})$.\\
 The  $ij$-th  coefficient of $G(f_1,\ldots,f_n)^{+} \,  (P+\delta K)$ is of the following form:
$$\sum_{l=1}^n g^{+}_{il} \, p_j +\delta \sum_{l=1}^ng^{+}_{il} \, k_{lj}=  p_j +\delta \sum_{l=1}^ng^{+}_{il} \, k_{lj},$$
 since $G(f_1,\ldots,f_n)^{+}$ is row stochastic.\\
Now, we denote by $c_{ij}$ the coefficients of $G(f_1,\ldots,f_n)^{+} \, K$. We have:
$$
p_j +\delta c_{ij} \geq\min(p_i) +\delta c_{ij}.
$$
Thus, if $\delta$ satisfies $\min(p_i) +\delta c_{ij} \geq 0$, we obtain the desired result.\\
This condition is trivially satisfied when $c_{ij} \geq 0$.\\
 In the general case, we have:
$\min(p_i) +\delta c_{ij} \geq 0 \iff \dfrac{min(p_i)}{|c_{ij}|}\geq \delta.$

Thus, if $\delta$ satisfies the bound given in the theorem, then all the coefficients of $G(f_1,\ldots,f_n)^{+} \,  (P+\delta K)$ are non-negative.\\

In order to simplify the notation, we denote by $G$ the Gram matrix $G(f_1,\ldots,f_n).$ In order to show the second item of the theorem, we use the well known equality satisfied by the pseudo-inverse $G^+$ : $G G^+ = \Pi_{\Ima(G)},$ where $\Pi_{\Ima(G)}$ is 
the orthogonal projection onto the image of $G$ (that we denoted by $\Ima(G)$). We are going to show the inclusion 
$\Ima(P + \delta K) \subset  \Ima(G)$ and this will imply the desired equality. 

Let $x=(x_1, \ldots, x_n)^T$ a vector in  $\ker(G)$. From Lemma \ref{lem:nullspace} we have \mbox{$\sum_{i=1}^n x_i \mu_i =0$} which implies  
$x^T K = 0$ since $K$ is a proper matrix.\\
 The matrix $G$ is symmetric, therefore $x^T G =0$ and $x^T G e= x^T e= 0$ which implies $x^T P = 0$.\\
This gives $x^T (P + \delta K) = 0$ and we conclude that $x \in \Ima(P + \delta K)^\perp $ (the orthogonal complement of $\Ima(P + \delta K)$). \\
 We have obtained the inclusion $\ker(G) \subset \Ima(P + \delta K)^\perp$.\\Thus $\ker(G)^{\perp} \supset \Ima (P+\delta K)$. Furthermore, we have $\ker(G)^{\perp}=\Ima(G^T)=\Ima(G)$, since $G$ is symmetric. This gives $\Ima(G) \supset \Ima(P+\delta K)$, and thus the second item is proved.\\

The third item is a direct consequence of the previous items and Lemma \ref{lem:stab-sto}.
\end{proof}

The previous theorem shows that if the measures are linearly independent then when can get an hyper envy-free fair division for all matrix $K=(k_{ij})$ such that $\sum_{j=1}^n k_{ij}=0$, see Definition \ref{def:propre}. Intuitively, if the measures are ``nearly" linearly dependent  then $\delta$ will be small. The following gives a precise statement for this intuition.\\

As we suppose that the measures are linearly independent, $G=G(f_1,\ldots,f_n)$ is inversible, see Lemma \ref{lem:nullspace}. We denote by $\tilde{g}_{ij}$ the coefficients of $G^{-1}$. Cramer's rule gives the following equality
$$\tilde{g}_{ij}=\dfrac{ \det(G_{ji}) }{ \det(G)  },$$
where $G_{ji}$ is the matrix where we have substituted in the matrix $G$ the $i$-th column by the column vector $e_j=(0,\ldots,0,1,0,\ldots,0)^T$.\\
Thus, if $\tilde{g}_{ij} \neq 0$ then we have 
$$0=\det(G)-\dfrac{1}{\tilde{g}_{ij}} \det(G_{ji}),$$
and by linearity of the determinant relatively to the $i$-th column we deduce:
$$0=\det\Big(G -\dfrac{1}{\tilde{g}_{ij}} E_{ji} \Big),$$
where $E_{ji}$ is the matrix with just one $1$ in the $i$-th column and $j$-th row and $0$ elsewhere.\\
We deduce that $G-1/\tilde{g}_{ij} \, E_{ij}$ is a singular matrix. We denote by $\Sigma_n$ the set of all $n$ by $n$ singular matrices.\\
It follows $\| 1/\tilde{g}_{ij} E_{ji} \| \geq d(G,\Sigma_n)$, where $\|.\|$ is a norm in the space of square matrices and $d(G,\Sigma_n)$ is the distance relatively to this norm  between the matrix $G$ and the set $\Sigma_n$. We suppose that $\|E_{ji}\|=1$. This assumption is not restrictive since all the classical norms satisfy this property. Thus
$$\dfrac{1}{|\tilde{g}_{ij}|} \geq d(G,\Sigma_n).$$
We deduce that the coefficients $c_{ij}$ of $G(f_1,\ldots,f_n)^{-1}K$ satisfy:
$$|c_{ij}|\leq n \, \underset{ij}{\max} |k_{ij}| \, \max_{ij} \tilde{g}_{ij} \leq \frac{n\, \max_{ij} |k_{ij}|}{d(G,\Sigma_n)}.$$
Thus $$\dfrac{\underset{i}{\min}(p_i)}{|c_{ij}|} \geq \underset{i}{\min}(p_i)\dfrac{d(G,\Sigma_n)}{n \, \underset{ij}{\max} |k_{ij}|}.$$

This gives the following corollary:

\begin{Cor}
Let $\|.\|$ be a norm in the space of $n\times n$ matrices such that $\|E_{ij}\|=1$ where $E_{ij}$ is the matrix with one $1$ in the $i$-th row  and $j$-th column and $0$ elsewhere. We denote by $d$ the associated distance.\\
 Let $\Sigma_n$ be the set of singular $n\times n$ matrices.\\
If the measures are linearly independent, then for all $p$ and all $K=(k_{ij})$ such that $\sum_{j=1}^n k_{ij}=0$, for all $i=1,\ldots,n$,  there exists an hyper envy-free fair division with $\delta$ satisfying the inequality:
$$\delta \leq \underset{i}{\min}(p_i) \, \dfrac{d(G,\Sigma_n)}{n\, \underset{ij}{\max} |k_{ij}|}.$$
\end{Cor}

This bound is smaller than the previous one but it gives a natural relation between the size of $\delta$, $p$, $K$, $n$ and the linear independency of the measures.\\
The coefficient $d(G,\Sigma_n)$ describes the linear independency of the measures. 
 Indeed, we use a Gram matrix because the density functions belong to an infinite dimensional vector space and thus we cannot express the determinant of these functions in a finite basis of the ambient space. The Gram matrix encodes all the linear relations between the measures. It is natural to control how the measures are ``linearly independent" in term of the distance between this matrix and the set of singular matrices $\Sigma_n$.\\

Our corollary can be used for all classical norms in the space of $n\times n$ matrices. The choice of a norm and then of a metric corresponds to different points of view.\\
As an example, if we consider the norm $\|.\|_{\infty}$ we could measure the quality of the corresponding  fair division in term of the distance $d_{\infty}(M,I)=\|M-I\|_{\infty}$, where $I$ is the identity matrix.  Indeed, $I$ corresponds to the case where each player thinks that he or she has the maximum of his or her utility function. If $M=\big(\mu_i(X_j)\big)$ is a sharing matrix,
then 
$$d_{\infty}(M,I)=\max_{i} \sum_{j=1}^n |\mu_i(X_j) -\delta_{ij}|.$$
Then $1- d_{\infty}(M,I)$  can be seen as a kind of  Rawlsian social welfare function.\\
Thus the choice of a metric is related to a social welfare function.

\subsection{Necessary conditions}$\,$\\

In order to make more natural the previous conditions on the matrix $K$, we give the following necessary condition:

\begin{Lem}
Suppose that $P+\delta K$ is a sharing matrix, where $P$ is  the matrix where each row is $(1/n,\ldots,1/n)$ and $\delta >0$. Then $K=(k_{ij})$ is a proper matrix. 
\end{Lem}
\begin{proof}
As $P$ is a stochastic matrix, and $P+ \delta K$ is a sharing matrix, we deduce that $\sum_{j=1}^n k_{ij}=0$. Now suppose that we have the following relation between the measure $\mu_i$: $\sum_{i=0}^n \lambda_i \mu_i=0$. This gives the following relation: 
$$(\sharp) \quad \sum_{i=0}^n \lambda_i (1/n+\delta k_{ij})=0,$$
because there exists measurable sets $X_j$ such that $\mu_i(X_j)=1/n+\delta k_{ij}$ since $P+\delta K$ is a sharing matrix. Furthermore, as $\sum_{i=0}^n \lambda_i \mu_i=0$, we get that $\sum_{i=1}^n \lambda_i \mu_i(X) =0$.  Thus $\sum_{i=1}^n \lambda_i =0$. As $\delta >0$, the equation $(\sharp)$ gives the desired result: 
\mbox{$\sum_{i=0} \lambda_ik_{ij}=0.$}
\end{proof}
\section{Effective methods}\label{sec:test}
In the following, we are going to study how to test if a fair division satisfying some inequalities constraints exists. Then we will show how to get this kind of fair division when the density functions are piecewise constant.
\subsection{How to test the existence of a fair division given by a relation matrix}
In this section, we use the previous theorem in order to study another fair division problem. Suppose that $p=(p_1,\ldots,p_n)$ is a point such that $p_1+\cdots+p_n=1$ with $p_i$ positive and $r_{ij}$ are $n^2$ relations in $\{ <, = , > \}$. Barbanel has asked the question:\emph{ How can we decide if there exists a partition $X=\sqcup_{i=1}^n X_i$ such that $\mu_i(X_j)\, r_{ij} \, p_j$?}\\

If there exists a proper matrix $K=(k_{ij})$ such that $k_{ij}\, r_{ij} \,0$ then by Theorem~\ref{thm1} there exists a partition $X=\sqcup_{i=1}^n X_i$ such that $\mu_i(X_j)\, r_{ij}\, p_j$.\\
 Conversely, if $X=\sqcup_{i=1}^n X_i$ is a partition such that $\mu_i(X_j)\, r_{ij} \, p_j$ then the matrix $K=(k_{ij})$ where $k_{ij}=\mu_i(X_j)- p_j$ is a proper matrix such that $k_{ij} \,r_{ij}\, 0$. Thus as already remarked by Barbanel, the previous question can be restated as: \emph{How can we decide if there exists a proper matrix $K$ satisfying $k_{ij}\, r_{ij} \,0$?}\\

Suppose that we know a basis $\{\lambda_1,\ldots,\lambda_l\}$ for the relations between the measures $\mu_i$, where $\lambda_\alpha=(\lambda_{\alpha 1},\ldots,\lambda_{\alpha n})$. Then, we just have to solve the following system:
$$
(\mathcal{S}) \quad 
\begin{cases}
 \forall i,\quad \sum_{j=1}^n k_{ij}=0,\\
\forall j, \forall \alpha, \quad \sum_{i=1}^n \lambda_{\alpha i} k_{ij}=0,  \\
\forall i, \forall j, \quad k_{ij}\, r_{ij} \, 0.
\end{cases}
$$

We thus have to solve a linear system of equalities and inequalities. This can be done for example by using the Fourier-Motzkin elimination method, see e.g. \cite{Kuhn_ineq}.\\

Now the question is : \emph{How can we find a basis $\{\lambda_1,\ldots,\lambda_l\}$ of relations between the $\mu_i$?}\\

If the measure $\mu_i$ are given thanks to a density function $\varphi_i(x)$ relatively to a given measure, e.g. the Lebesgue measure, then we can answer the previous question.\\
Indeed, we can set 
$$f_i(x)=\dfrac{\varphi_i(x)}{\sum_{i=1}^n \varphi_i(x)},$$ for $x$ such that $\sum_{i=1}^n \varphi_i(x)\neq 0$ and $f_i(x)=0$ otherwise. \\
Indeed, the density of $\mu=\mu_1+\cdots+\mu_n$ is $\varphi_1(x)+\cdots+\varphi_n(x)$. \\
Thus, if $\varphi_i$ are given explicitly then we can deduce an explicit formula for the $f_i$. This gives
$$\int_X f_i\,f_jd\mu=\int_X \dfrac{\varphi_i(x)\,\varphi_j(x)}{\sum_{i=1}^n \varphi_i(x)}dx.$$ 
If these expressions are computable then we have an explicit form for $G(f_1,\ldots,f_n)$. Thus in this situation the basis $\{\lambda_1,\ldots,\lambda_l\}$ can be computed easily. Indeed, this basis is a basis of $\ker G(f_1,\ldots,f_n)$ thanks to  Lemma \ref{lem:nullspace}. In conclusion, the system $\mathcal{S}$ can be solved with linear algebra only.

\begin{Rem} As mentioned before, the previous strategy gives a theoretical method for solving Barbanel's problem. Our solution do not give an algorithm in the Robertson-Webb model, see \cite{RobertsonWebb}, \cite[Chapter 13]{handbook}. However, if $X=[0,1]$ and the $\varphi_i$ are polynomials or simple functions i.e. linear combination of indicator functions, then $G(f_1,\ldots,f_n)$ can be constructed effectively and our method can be used.
\end{Rem}

\begin{Ex}\label{example}
Suppose that $X=[0,1]$ and that the three measures are given with their densities $\varphi_1(x)=10\rchi_{[0,1/10]}(x)$, $\varphi_2(x)=\frac{10}{9}\rchi_{[1/10,1]}(x)$, $\varphi_3(x)=\rchi_{[0,1]}(x)$. Suppose also that $p=(1/3,1/3,1/3)$ and that $R_1=(r_{ij}^{(1)})$ $R_1=(r_{ij}^{(2)}) $ are the following matrix
$R_1=\begin{pmatrix}
>&=&<\\
> & >& <\\
<&<&>
\end{pmatrix}$,  
 $R_2=\begin{pmatrix}
>&=&<\\
< & >& >\\
<&>&>
\end{pmatrix}$. \\
As we can compute easily the expression $\displaystyle \int_X \dfrac{\varphi_i(x)\varphi_j(x)}{\sum_{i=1}^n \varphi_i(x)}dx.$
The Gram matrix associated to this situation is:
$$G(f_1,f_2,f_3)=
\begin{pmatrix}
\dfrac{10}{11}&0&\dfrac{1}{11}\\
&&\\
0&\dfrac{10}{19}&\dfrac{9}{19}\\
&&\\
\dfrac{1}{11}&\dfrac{9}{19}&\dfrac{91}{209}
\end{pmatrix}.
$$
\end{Ex}
We can remark that this matrix is a sharing matrix corresponding to a proportinal fair division as proved in Remark \ref{rem_prop}, but it is not an envy-free fair division.\\

A basis for the kernel of this matrix is given by $\{(1,9,-10)\}$. Thus the only relation between the measure is: $\mu_1+9\mu_2=10\mu_3$.

Now, we have to solve the following system when $s=1$ or $s=2$:
$$(\mathcal{S}_s) \quad \begin{cases}
 \forall i,\quad \sum_{j=1}^3 k_{ij}=0,\\
\forall j,  \quad \sum_{i=1}^3 \lambda_{i} k_{ij}=0,  \\
\forall i, \forall j, \quad k_{ij}\, r_{ij}^{(s)} \, 0,
\end{cases}
$$
where $(\lambda_1,\lambda_2,\lambda_3)=(1;9;-10)$.\\
When $s=1$ the system have no solution. Indeed, we must have $k_{11}>0$, $k_{21}>0$, $k_{11}+9k_{21}=10 k_{31}$, and $k_{31}<0$ which gives a contradiction.\\
When $s=2$ we can compute solutions. For example the following matrix $K=(k_{ij})$ is a solution of $(\mathcal{S}_2)$:
$$K=
\begin{pmatrix}
1&0&-1\\
&&\\
-\dfrac{1}{3}&\dfrac{1}{9}&\dfrac{2}{9}\\
&&\\
-\dfrac{2}{10}&\dfrac{1}{10}&\dfrac{1}{10}
\end{pmatrix}.
$$
\

\subsection{The case of piecewise constant density functions}
In this last section, we consider the case where the density functions are piecewise constant. This means $\varphi_i(x)=\sum_j c_{ij}\rchi_{I_{ij}}(x)$, where $\rchi_{I_{ij}}$ is the indicator function of the interval $I_{ij}=[x_{ij},x_{i(j+1)}]$, $x_{ij}\in [0,1]$.\\
Suppose that the measures $\mu_i$ are given by these density functions and that we want to construct an hyper envy-free fair division relatively to a given matrix $K=(k_{ij})$ and $p=(p_1,\ldots,p_n)$.\\
We denote by $\mathcal{I}$ the set of intervals $[x_{kl},x_{mn}]$ such that there exists no $x_{ij}$ satisfying $x_{kl}< x_{ij} < x_{mn}$. This means that on each interval $I \in \mathcal{I}$ the density functions are constant.\\

Now, we explain how to construct an hyper envy-free fair division:\\
Since all the densities $\varphi_i$ are constant on $I$ we have for all subintervals $J$ of $I$: 
$$\mu_i(J)\ell(I)=\ell(J)\mu_i(I),$$
where $\ell(I)$ is the length of $I$ for the Lebesgue measure.

Thus if $\mu_i(I)\neq 0$  we have:
$$\dfrac{\mu_i(J)}{\mu_i(I)}=\dfrac{\ell(J)}{\ell(I)}.$$

Now, consider $\alpha_{1,I}, \ldots,\alpha_{n,I}$ such that $\alpha_{i,I} \geq 0$ and $\sum_i \alpha_{i,I}=1$. We can divide each $I \in \mathcal{I}$ in order to have a partition of $I$ into $n$ subintervals $I_1, \ldots,I_n$ satisfying $\ell(I_j)=\alpha_{j,I}\ell(I)$. The previous relation implies
$$(\flat) \quad \mu_i(I_j)=\alpha_{j,I}\mu_i(I) \, \textrm{ for all } i,j.$$

%
%
%
%
%
%
%
%
%
%
%
%
%
%
%

Thus we set $X_j= \sqcup_{I \in \mathcal{I}} I_j$, and we get:
$$X=\sqcup_{j=1}^n X_j \textrm{ and } \mu_i(X_j)=\sum_{I \in \mathcal{I}} \alpha_{j,I}\mu_i(I).$$


Now, in order to find an hyper envy-free fair division relatively to $K$ and $p$ we have to solve:
$$(\mathcal{S}')
\quad
\begin{cases}
\forall I \in \mathcal{I}, \quad \sum_{i=1}^n \alpha_{i,I}=1,\\
\forall i \in \{1,\ldots,n\}, \forall I \in \mathcal{I}, \quad \alpha_{i,I} \geq 0,\\
\forall i,j \in \{1,\ldots,n\}, \quad \sum_{I \in \mathcal{I}} \alpha_{j,I} \mu_i(I)=p_j +k_{ij}\delta,\\
\delta > 0.
\end{cases}
$$

\begin{Ex}\label{exsuite}
With the density functions and the matrix $K=(k_{ij})$  of  Example~\ref{example}, and $p=(1/3,1/3,1/3)$, we have $\mathcal{I}=\{[0,1/10],[1/10,1]\}$ and  the system $(\mathcal{S}')$ gives:
$\begin{cases}
 \sum_{i=1}^3 \alpha_{i,[0,1/10]}=1\\
\forall i \in \{1,\ldots,3\},\quad \alpha_{i,[0,1/10]} \geq 0,\\
 \sum_{i=1}^3 \alpha_{i,[1/10,1]}=1\\
\forall i \in \{1,\ldots,3\},\quad \alpha_{i,[1/10,1]} \geq 0,\\
\forall i,j \in \{1,\ldots,3\}, \quad \sum_{I \in \mathcal{I}} \alpha_{j,I} \mu_i(I)=1/3 +k_{ij}\delta,\\
\delta>0.
\end{cases}
$
\end{Ex}
We get the following solution:\\
We set $a:=\alpha_{1,[0,1/10]}$ and $a$ must satisfies the condition 
$a \in ]1/3,2/3[$.\\
Then $\alpha_{2,[0,1/10]}=1/3$, $\alpha_{3,[0,1/10]}=2/3-a$, $\alpha_{1,[1/10,1]}=4/9-1/3a$,\\ $\alpha_{2,[1/10,1]}=8/27+1/9a$, $\alpha_{3,[1/10,1]}=7/27+2/9a$, \mbox{$\delta=a-1/3$}.
Thus, with $\alpha=1/2$ we get:\\
$a:=\alpha_{1,[0,1/10]}=1/2$, $\alpha_{2,[0,1/10]}=1/3$,  $\alpha_{3,[0,1/10]}=1/6$, $\alpha_{1,[1/10,1]}=5/18$, $\alpha_{2,[1/10,1]}=19/54$, $\alpha_{3,[1/10,1]}=10/27$, \mbox{$\delta=1/6$}.\\
We divide $[0,1/10]$ and $[1/10,1]$ in order to respect the condition $(\flat)$. For the interval $I=[0,1/10]$ we get: $I_1=[0,1/20]$, $I_2=[1/20,1/20+1/3\times 1/10]=[1/20,1/12]$ and $I_3=[1/12,1/10]$. With the same method we divide $[1/10,1]$ and then we get:
$$X_1=\left[0,\dfrac{1}{20}\right] \sqcup \left[\dfrac{1}{10},\dfrac{7}{20} \right], \, X_2=\left[\dfrac{1}{20},\dfrac{1}{12}\right] \sqcup \left[\dfrac{7}{20},\dfrac{2}{3}\right], \, X_3=\left[\dfrac{1}{12},\dfrac{1}{10}\right] \sqcup \left[\dfrac{2}{3},1\right].$$
By construction, this partition gives an hyper envy-free fair division relatively to $K$ and $p=(1/3,1/3,1/3)$ with $\delta=1/6$. Indeed, we can check that the sharing matrix associated to the previous partition is
$$\begin{pmatrix}
\dfrac{1}{2}&\dfrac{1}{3}&\dfrac{1}{6}\\
&&\\
\dfrac{5}{18}& \dfrac{19}{54}&\dfrac{10}{27}\\
&&\\
\dfrac{3}{10}& \dfrac{7}{20}&\dfrac{7}{20}
\end{pmatrix}.
$$

Therefore, this partition gives also a fair division respecting the inequalities given in the matrix $R_2$.\\

In conclusion, when the density functions are piecewise constant then  we have a method to compute an hyper envy-free fair division. Thus we can : decide if a fair division given by a relation matrix exists, and then construct a partition giving this fair division.\\

\begin{Rem}
The pseudo-inverse of the matrix $G(f_1,f_2,f_3)$  obtained in Example~\ref{example} and Example~\ref{exsuite} is 
$$G(f_1,f_2,f_3)^+=
\begin{pmatrix}
\frac{36191}{33124}&-{\frac{3519}{33124}
}&{\frac{113}{8281}}\\ 
&&\\
-{\frac{3519}{33124}}&{\frac{19471}{33124}}&{\frac{4293}{8281}}\\
&&\\
 \frac{113}{8281}&{\frac{4293}{8281}}&\frac{3875}{8281} 
\end{pmatrix}.
$$
The matrix $G(f_1,f_2,f_3)^+K$ is
$$G(f_1,f_2,f_3)^+K=
\begin{pmatrix}
{\frac{512}{455}}&-{\frac{19}{1820}}&-{\frac{2029}{1820}}\\ 
&&\\
-{\frac{554}{1365}}&{\frac{1919}{16380}}&{\frac{4729}{16380}}\\ 
&&\\
-{\frac{23}{91}}&{\frac{19}{182}}&{\frac{27}{182}}
\end{pmatrix}.
$$
Then $\underset{ij}{\max}|\big(G(f_1,\ldots,f_n)^{+} K\big)_{ij}|=512/455$.\\
 As $\underset{i}{\min}(p_i)=1/3$, we get $\mathcal{B}=455/1536$. \\
 Thus Theorem \ref{thm1} gives the existence of an hyper envy-free fair division relatively to $K$ and $p$ with $\delta < 455/1536 \approx 0.296$ without computing a solution of $(\mathcal{S}')$.\\
  In the above example, the solution of $(\mathcal{S}')$ gives $\delta=a-1/3$ with $a<2/3$. Thus in this particular situation an hyper envy-free fair division exists with $\delta<1/3$. This shows that our bound $\mathcal{B}$ is not optimal in this case, but the order of magnitude of $\mathcal{B}$ is not too small.
\end{Rem}

\section*{Conclusion}
The key ingredient of this article is the Gram matrix $G(f_1,\ldots,f_n)$. Thanks to this matrix we can construct the sharing matrix $P+\delta K$, see Theorem \ref{thm1}, and we can also compute the relations between the measures, see Lemma \ref{lem:nullspace}.\\
This matrix seems very useful and we have already noticed that this matrix corresponds to a proportional and symmetric (i.e. $\mu_i(X_j)=\mu_j(X_i)$) fair division, see Remark \ref{rem_prop}. Thus  a natural question appears: \\
\emph{Is $G(f_1,\ldots,f_n)$ the optimal fair division for a certain criterion and what is the meaning of this criterion in terms of fair division?}\\

\emph{Acknowledgments}: The authors thank L.-M. Pardo for fruitful discussions during the preparation of this article.

  \bibliographystyle{alpha} 
 
\bibliography{cakebiblio}

\end{document}